\documentclass[conference]{ieeeconf}
\IEEEoverridecommandlockouts

\usepackage{cite}
\usepackage[dvipsnames]{xcolor}
\usepackage{algorithm}
\usepackage[noend]{algpseudocode}

\usepackage{graphicx}
\usepackage{textcomp}
\usepackage{booktabs}
\usepackage{arydshln}
\def\BibTeX{{\rm B\kern-.05em{\sc i\kern-.025em b}\kern-.08em T\kern-.1667em\lower.7ex\hbox{E}\kern-.125emX}}
\usepackage{multirow}

\usepackage[cmex10]{amsmath}
\usepackage{amssymb, mathrsfs}
\usepackage{mathtools}
\usepackage{microtype}
\usepackage{enumerate, paralist}

\usepackage{afterpage}

\usepackage{amsthm}
\newtheorem{proposition}{Proposition}
\newtheorem{lemma}[proposition]{Lemma}
\newtheorem{theorem}[proposition]{Theorem}

\theoremstyle{definition}

\newtheorem{assumption}[proposition]{Assumption}
\theoremstyle{remark}

\newcommand{\oprocendsymbol}{\rule{.7em}{.7em}}
\newcommand{\oprocend}{\relax\ifmmode\else\unskip\hfill\fi\oprocendsymbol}

\newcommand{\col}{\mathrm{col}}
\newcommand{\Diag}{\mathrm{Diag}}
\newcommand{\mat}[1]{\begin{bmatrix}#1\end{bmatrix}}
\newcommand{\neatmat}[1]{\text{\small$\mat{#1}$}}
\newcommand{\compactmat}[1]{\text{\footnotesize$\mat{#1}$}}
\newcommand{\smallmat}[1]{\begin{bsmallmatrix}#1\end{bsmallmatrix}}
\newcommand{\compact}[1]{\text{\footnotesize$#1$}}
\newcommand{\Hank}{\mathcal{H}}
\newcommand{\Toep}{\mathrm{Toep}}

\newcommand{\real}{\mathbb{R}}
\newcommand{\integer}{\mathbb{Z}}
\renewcommand{\natural}{\mathbb{N}}
\newcommand{\symmetric}{\mathbb{S}}
\newcommand{\wideminus}{{\scalebox{1.5}[1] -}}
\newcommand{\transpose}{{\sf T}}
\newcommand{\auxiliary}{{\sf aux}}
\newcommand{\original}{{\sf orig}}
\newcommand{\initial}{{\sf ini}}
\newcommand{\data}{{\rm d}}
\newcommand{\dash}{\text{-}}
\usepackage{accents}
\newcommand\thickbar[1]{\accentset{\rule{.4em}{.8pt}}{#1}}

\DeclareMathOperator*{\minimize}{minimize}

\newcommand{\tb}{\color{blue}}
\newcommand{\tgrey}{\color{black!40}}

\renewcommand{\tb}{\color{black}}
\renewcommand{\tgrey}{\color{black}}

\newcommand{\Versions}[2]{#2}

\begin{document}

\title{Distributionally Robust Stochastic Data-Driven Predictive Control \\ with Optimized Feedback Gain}

\author{Ruiqi Li,  John W. Simpson-Porco, and Stephen L.\ Smith
\thanks{This research is supported in part by the Natural Sciences and Engineering Research Council of Canada (NSERC).}%
\thanks{Ruiqi Li and Stephen L. Smith are with the Electrical and Computer Engineering at the University of Waterloo, Waterloo, ON, Canada
{\tt\small \{r298li,stephen.smith\}@uwaterloo.ca}}%
\thanks{John W. Simpson-Porco is with the Department of Electrical and Computer Engineering at the University of Toronto, Toronto, ON, Canada
{\tt\small jwsimpson@ece.utoronto.ca}}
}

\maketitle
\thispagestyle{empty}
\pagestyle{empty}

\setlength{\abovedisplayskip}{.5em}
\setlength{\belowdisplayskip}{.5em}
\begin{abstract}
    We consider the problem of direct data-driven predictive control for unknown stochastic linear time-invariant (LTI) systems with partial state observation. Building upon our previous research on data-driven stochastic control, this paper (i) relaxes the assumption of Gaussian process and measurement noise, and (ii) enables optimization of the gain matrix within the affine feedback policy. Output safety constraints are modelled using conditional value-at-risk, and enforced in a distributionally robust sense. Under idealized assumptions, we prove that our proposed data-driven control method yields control inputs identical to those produced by an equivalent model-based stochastic predictive controller. A simulation study illustrates the enhanced performance of our approach over previous designs.
\end{abstract}

\section{Introduction}

Model predictive control (MPC) is a widely used technique for multivariate control \cite{MPC:mayne2014}, adept at handling constraints on inputs, states and outputs while optimizing complex performance objectives.
\Versions{}{Constraints typically model actuator limits, or encode safety constraints in safety-critical applications, and }MPC employs a system model to predict how inputs influence state evolution.
\Versions{W}{Both deterministic and stochastic frameworks have been developed to account for plant uncertainty in MPC. While \emph{Robust MPC} \cite{RMPC:bemporad2007} approaches model uncertainty in a worst-case deterministic sense, w}ork on \emph{Stochastic MPC (SMPC)} \cite{SMPC:mesbah2016} has focused on describing model uncertainty probabilistically.
SMPC methods optimize over feedback control policies rather than control actions\Versions{}{, resulting in performance benefits when compared to the na\"ive use of deterministic MPC \cite{SMPC_MPC:kumar2019},}
and accommodate probabilistic and risk-aware constraints.

The system model required by MPC (and SMPC) is sometimes obtained from identification, making MPC an \emph{indirect} design method, since one goes from data to a controller through an intermediate modelling step \cite{DDC:dorfler2023}. In contrast, data-driven or \emph{direct} methods seek to compute controllers directly from input-output data\Versions{.}{, showing promise for complex or difficult-to-model systems \cite{DDC:Hou2013}.}
Accounting for constraints in control, \emph{Data-Driven Predictive Control (DDPC)} methods were developed, including Data-Enabled Predictive Control (DeePC) \cite{DeePC:coulson2019a, DeePC:coulson2019b, DeePC:coulson2021} and Subspace Predictive Control (SPC) \cite{SPC:huang2008}\Versions{}{, both of which have been applied in multiple experiments \cite{DeePCApp:elokda2021quadcopters, DeePCApp:carlet2020motorDrives, DeePCApp:huang2021oscillationDamping}}.
For \emph{deterministic} LTI systems in theory, both DeePC and SPC produce equivalent control actions as from MPC. 

Real-world systems often deviate from idealized deterministic LTI models, exhibiting stochastic and non-linear behavior, with noise-corrupted data.
To address these challenges, data-driven methods must account for noisy data and measurements. 
For instance, in SPC applications, required predictor matrices are often computed using denoising techniques \cite{SPC:huang2008}.
Regularized and distributionally robust DeePC were also developed for stochastic systems \cite{DeePC:coulson2019a, DeePC:coulson2019b, DeePC:coulson2021}.
Unlike in the deterministic case, however, these stochastic adaptations of DeePC and SPC lack theoretical equivalence to model-based MPC.

Recognizing this gap, some recent advancements in DDPC have aimed to establish equivalence with MPC methods for stochastic systems.
The work in \cite{PCE:pan2022b, PCE:pan2022a} proposed a DDPC framework for stochastic systems, and their method performs equivalently to SMPC if stochastic signals can be exactly expressed by their Polynomial Chaos Expansion.
This paper builds in particular on our previous work \cite{SDDPC}, where we proposed a data-driven control method for stochastic systems, {\tb without estimation of disturbance as required in \cite{PCE:pan2022a, PCE:pan2022b}}, and established that the method has equivalent control performance to SMPC when offline data is noise-free.
\Versions{{\tb An extended version of the paper can be found in \cite{EXTENDED}.}}{}

\emph{Contribution:} This paper contributes towards the continued development of high-performance DDPC methods for stochastic systems. Specifically, in this paper we develop a stochastic DDPC strategy utilizing distributionally robust conditional value-at-risk constraints, providing an improved safety constraint description when compared to our prior work in~\cite{SDDPC}, and providing robustness against non-Gaussian (i.e., possibly heavy-tailed) process and measurement noise. Additionally, in contrast with the fixed feedback gain in~\cite{SDDPC}, we consider control policies where feedback gains are decision variables in the optimization, giving a more flexible parameterization of control policies.
As theoretical support for the approach, under technical conditions, we establish equivalence between our proposed design and a corresponding SMPC. Finally, a simulation case study compares and contrasts our design with other recent stochastic and data-driven control strategies.

\emph{Notation:} 
Let $M^\dagger$ be the pseudo-inverse of a matrix $M$. Let $\otimes$ denote the Kronecker product. Let $\symmetric^q_+$ (resp. $\symmetric^q_{++}$) be the set of $q\times q$ positive semi-definite (resp. definite) matrices. 
Let $\col(M_1, \ldots, M_k)$ (resp. $\Diag(M_1, \ldots, M_k)$) denote the vertical (resp. diagonal) concatenation of matrices / vectors $M_1, \ldots, M_k$.
Let $\integer_{[a,b]} := [a,b] \cap \integer$ denote a set of consecutive integers from $a$ to $b$, and let $\integer_{[a,b)} := \integer_{[a,b-1]}$.
For a $\real^q$-valued discrete-time signal $z_t$ with integer index $t$, let $z_{[t_1, t_2]}$ denote either a sequence $\{z_t\}_{t=t_1}^{t_2}$ or a concatenated vector $\col(z_{t_1}, \ldots, z_{t_2}) \in \real^{q(t_2-t_1+1)}$ where the usage is clear from the context; similarly, let $z_{[t_1,t_2)} := z_{[t_1,t_2-1]}$.
A matrix sequence $\{M_t\}_{t=t_1}^{t_2}$ and a function sequence $\{\pi_t(\cdot)\}_{t=t_1}^{t_2}$ are denoted by $M_{[t_1,t_2]}$ and $\pi_{[t_1,t_2]}$ respectively.

\section{Problem Setup} \label{SECTION:problem}

Consider a stochastic linear time-invariant (LTI) system
\begin{align} \label{Eq:LTI}
    x_{t+1} = A x_t + B u_t + w_t, \quad
    y_t = C x_t + D u_t + v_t
\end{align}
with input $u_t \in \real^m$, state $x_t \in \real^n$,  output $y_t \in \real^p$, process noise $w_t \in \real^n$, and measurement noise $v_t \in \real^p$. 
The system $(A,B,C,D)$ is assumed as a minimal realization, but the matrices themselves are \emph{unknown} and the state $x_t$ is \emph{unmeasured}; we have access only to the input $u_t$ and output $y_t$ in \eqref{Eq:LTI}.
The probability distributions of $w_t$ and $v_t$ are \emph{unknown}, but we assume that $w_t$ and $v_t$ have zero mean and zero auto-correlation (white noise), are uncorrelated, and their variances $\Sigma^{\sf w} \in \symmetric^n_+$ and $\Sigma^{\sf v} \in \symmetric^p_{++}$ are known. The initial state $x_0$ has given mean $\mu^{\sf x}_\initial$ and variance $\Sigma^{\sf x}$ and is uncorrelated with the noise. We record these conditions as
\begin{align} 
\label{Eq:noise_meanvar} 
    & \mathbb{E} \big[ \neatmat{w_t \\[-.25em] v_t} \big] = 0, \quad
    \mathbb{E} \big[ \neatmat{w_t \\[-.25em] v_t} \neatmat{w_s \\[-.25em] v_s} \!\raisebox{.25em}{${}^\transpose$} \big] = \neatmat{\delta_{ts} \Sigma^{\sf w} & 0 \\[-.25em] 0 & \delta_{ts} \Sigma^{\sf v}}, \\
\label{Eq:initial_state_condition} 
    & \mathbb{E}[x_0] = \mu^{\sf x}_\initial, \quad
    \mathrm{Var}[x_0] = \Sigma^{\sf x}, \quad
    \mathbb{E} \big[ x_0 \neatmat{w_t \\[-.25em] v_t} \!\raisebox{.25em}{${}^\transpose$} \big] = 0,
\end{align}
with $\delta_{ts}$ the Kronecker delta. {\tb Assume $(A, \Sigma^{\sf w})$ is stabilizable.}

In a reference tracking control problem for \eqref{Eq:LTI}, the objective is for the output $y_t$ to follow a specified reference signal $r_t \in \real^p$. The trade-off between tracking error and control effort may be encoded in an instantaneous cost
\begin{align} \label{Eq:stage_cost}
    J_t (u_t, y_t) := \Vert y_t-r_t \Vert_Q^2 + \Vert u_t \Vert_R^2 
\end{align}
to be minimized over a time horizon, with user-selected parameters $Q \in \symmetric^p_+$ and $R \in \symmetric^m_{++}$. This tracking should be achieved subject to constraints on the inputs and outputs. 
We consider here polytopic constraints, which in a deterministic setting would take the form $E \neatmat{u_t\\[-.25em] y_t} \leq f$ for all $t \in \natural_{\geq 0}$, and for some fixed matrix $E \in \real^{q \times (m+p)}$ and vector $f \in \real^q$. We can equivalently express these constraints as the single constraint $h(u_t, y_t) \leq 0$, where
\begin{align} \label{Eq:constraint_function}
    h(u_t, y_t) := {\max}_{\, i \in \{1,\ldots,q\}} \; e^\transpose_i \neatmat{u_t \\[-.25em] y_t} - f_i,
\end{align}
with $e_i \in \real^{m+p}$ the transposed $i$-th row of $E$ and $f_i \in \real$ the $i$-th entry of $f$. For the system \eqref{Eq:LTI} which is subject to (possibly unbounded) stochastic disturbances, the deterministic constraint $h(u_t,y_t) \leq 0$ must be relaxed. 
Beyond a traditional chance constraint $\mathbb{P}[h(u_t, y_t) \leq 0] \geq 1 - \alpha$ with a violation probability $\alpha \in (0,1)$, a \emph{conditional value-at-risk (CVaR)} constraint is more conservative; the CVaR at level $\alpha$ of $h(u_t,y_t)$ is defined as the expected value of $h(u_t,y_t)$ in the $\alpha \cdot$ 100\% worst cases, and takes extreme violations into account. With the noise distributions unknown, we must further guarantee satisfaction of the CVaR constraint for all possible distributions under consideration.
Let $\mathbb{D}$ denote a joint distribution of all random variables in \eqref{Eq:LTI} satisfying \eqref{Eq:noise_meanvar} and \eqref{Eq:initial_state_condition}, and let the \emph{ambiguity set} $\mathcal{D}$ be the set of all such distributions.
The \emph{distributionally robust CVaR (DR-CVaR)} constraint \cite{VanParys2015, Zymler2013} is then
\begin{align} \label{Eq:DR_CVaR_Constraint}
    {\sup}_{\, \mathbb{D} \in \mathcal{D}} \; \mathbb{D}\dash\mathrm{CVaR}_\alpha [ h(u_t, y_t) ] \leq 0,
\end{align}
where $\mathbb{D}\dash\mathrm{CVaR}_\alpha [z]$ is the CVaR value of a random variable $z \in \real$ at level $\alpha$ given distribution $\mathbb{D}$.

If the system matrices $A, B, C, D$ were known, this constrained tracking control problem subject to \eqref{Eq:DR_CVaR_Constraint} can be approached using SMPC, as described in Section \ref{SECTION:theory:SMPC}. Our objective is to develop a data-driven control method that produces equivalent control inputs as produced by SMPC. 

\section{Stochastic Model-Based and Data-Driven Predictive Control} \label{SECTION:theory}

We introduce a model-based SMPC framework in Section \ref{SECTION:theory:SMPC} and propose a data-driven control method in Section \ref{SECTION:theory:SDDPC}, with their theoretical equivalence in Section \ref{SECTION:theory:equivalence}. 


\subsection{A framework of Stochastic Model Predictive Control} \label{SECTION:theory:SMPC}

We focus here on output-feedback SMPC \cite{OFSMPC:farina2015, OFSMPC:joa2023, OFSMPC:ridderhof2020} which typically combines state estimation and feedback control. 
The formulation here broadly follows our prior work \cite{SDDPC}, but we now consider a DR-CVaR constraint in place of chance constraints, and we will allow optimization over the feedback gain. 
This SMPC scheme merges the established works on DR constrained control \cite{VanParys2015, Zymler2013} and output-error feedback \cite{Goulart2007},  while the combined framework is part of our contribution.


\subsubsection{State Estimation}
SMPC follows a receding-horizon strategy and makes decisions for $N$ upcoming steps at each \emph{control step}.
At control step $t = k$, we begin with prior information of the mean and variance of state $x_k$, namely
\begin{align} \label{Eq:control_step_state_meanvar}
    \mathbb{E}[x_k] = \mu^{\sf x}_k, \qquad
    \mathrm{Var}[x_k] = \Sigma^{\sf x},
\end{align}
where the mean $\mu^{\sf x}_k$ is computed from a state estimator to be described next; at the initial step $k=0$, $\mu^{\sf x}_0 = \mu^{\sf x}_\initial$ is a given parameter as in \eqref{Eq:initial_state_condition}.
{\tb For simplicity of computation, we let $\Sigma^{\sf x}$ in \eqref{Eq:initial_state_condition} and \eqref{Eq:control_step_state_meanvar} be the steady-state variance through the Kalman filter, as the unique positive semi-definite solution to the associated discrete-time algebraic Riccatti equation (DARE) \eqref{Eq:Kalman_gain:state_variance}, with observer gain $L_{\sf L} \in \real^{n \times p}$ in \eqref{Eq:Kalman_gain:gains}.}
\begin{subequations} \tb \label{Eq:Kalman_gain} \begin{align}
\label{Eq:Kalman_gain:state_variance}
    & \Sigma^{\sf x} = (A - L_{\sf L} C) \Sigma^{\sf x} A^\transpose + \Sigma^{\sf w} \\
\label{Eq:Kalman_gain:gains}
    & L_{\sf L} := A \Sigma^{\sf x} C^\transpose(C \Sigma^{\sf x} C^\transpose + \Sigma^{\sf v})^{-1}
\end{align} \end{subequations}
Estimates $\hat x_t$ of future states over the desired horizon are computed through the observer, {\tb with \emph{innovation} $\nu_t \in \real^p$,}
\begin{subequations} \label{Eq:state_estimation} \begin{align}
    \label{Eq:state_estimation:innovation}
    \nu_t &:= y_t - C \hat x_t - D u_t, && t \in \integer_{[k,k+N)} \\
    \label{Eq:state_estimation:iteration}
    \hat x_{t+1} &:= A \hat x_t + B u_t + L_{\sf L} \nu_t, && t \in \integer_{[k,k+N)}
    \\
    \label{Eq:state_estimation:initial}
    \hat x_k &:= \mu^{\sf x}_k
\end{align} \end{subequations}
{\tb where we utilize in \eqref{Eq:state_estimation:iteration} the observer gain $L_{\sf L}$ in \eqref{Eq:Kalman_gain:gains} so that \eqref{Eq:state_estimation} is equivalent to the steady-state Kalman filter.}
\Versions{}{{\tgrey While the noise here is potentially non-Gaussian, the Kalman filter is the best affine state estimator in the mean-squared-error sense, regardless of the distributions of $x_k, w_t, v_t$ once their means and variances are specified as in \eqref{Eq:noise_meanvar} and \eqref{Eq:control_step_state_meanvar} \cite[Sec. 3.1]{Humpherys2012}.}}

At the control step with condition \eqref{Eq:control_step_state_meanvar}, we can predict future states and outputs by simulating the noise-free model,
\begin{subequations} \label{Eq:nominal_state} \begin{align}
    \label{Eq:nominal_state:recursion}
    \thickbar x_{t+1} &:= A \thickbar x_t + B \thickbar u_t, & t \in \integer_{[k,k+N)} \\
    \label{Eq:nominal_state:output}
    \thickbar y_t &:= C \thickbar x_t + D \thickbar u_t, & t \in \integer_{[k,k+N)} \\
    \label{Eq:nominal_state:initial}
    \thickbar x_k &:= \mu^{\sf x}_k
\end{align} \end{subequations}
with \emph{nominal inputs} $\thickbar u_t$ as decision variables to be optimized, and with resulting \emph{nominal states} $\thickbar x_t$ and \emph{nominal outputs} $\thickbar y_t$.


\subsubsection{Feedback Control Policies}
\Versions{}{Our prior work \cite{SDDPC} was based on an affine feedback policy $u_t = \thickbar u_t - K (\hat x_t - \thickbar x_t)$ with a fixed feedback gain $K$.
Here we investigate control policies where the feedback gain is a time-varying decision variable. However, the naive parameterization
\begin{align} \label{Eq:feedback_policy_output_feedback}
    u_t \gets \thickbar u_t - K_t (\hat x_t - \thickbar x_t)
\end{align}
leads to non-convex bilinear terms of the decision variables $\thickbar u$ and $K_t$, as $\hat x_t, \thickbar x_t$ depend on $\thickbar u_{[k,t)}$ via \eqref{Eq:state_estimation}, \eqref{Eq:nominal_state}.}
\Versions{While \cite{SDDPC} considered an affine feedback policy with fixed feedback gain, here we apply an \emph{output error feedback} control policy \cite{Goulart2007}}{Thus, we instead apply an \emph{output error feedback} control policy \cite{Goulart2007}}
\begin{align} \label{Eq:feedback_policy}
    u_t \gets \pi_t ( \nu_{[k,t)} ) := \thickbar u_t + \textstyle{\sum_{s=k}^{t-1}} \, M^s_t \, \nu_s
\end{align}
where the nominal input $\thickbar u_t$ and feedback gains $M^s_t \in \real^{m\times p}$ are both decision variables, with innovation $\nu$ in \eqref{Eq:state_estimation:innovation}.
\Versions{}{The policy parameterization \eqref{Eq:feedback_policy} contains within it the policy \eqref{Eq:feedback_policy_output_feedback} as a special case: indeed, for a sequence of gains $K_{[k,k+N)}$, the selection {\tb for all $s,t \in \integer_{[k,k+N)}, s \leq t$}
\begin{align*}
    M^s_t \gets (A - B K_{t-1}) (A - B K_{t-2}) \cdots (A - B K_s) L_{\sf L}
\end{align*}
reduces \eqref{Eq:feedback_policy} to \eqref{Eq:feedback_policy_output_feedback}.}
Crucially, \eqref{Eq:feedback_policy} leads to jointly convex optimization in decision variables $\thickbar u, M^s_t$, as we will see next.

With the estimator \eqref{Eq:state_estimation} and policy \eqref{Eq:feedback_policy}, both input $u_t$ and output $y_t$ of \eqref{Eq:LTI} can be written as affine functions of the decision variables, through direct calculation, with $\thickbar y$ in \eqref{Eq:nominal_state}, 
\begin{align} \label{Eq:IO_wrt_independent_RV}
    \neatmat{u_t \\[-.25em] y_t} = \neatmat{\thickbar u_t \\[-.25em] \thickbar y_t} + \Lambda_t \, \eta_k, \quad t \in \integer_{[k,k+N)},
\end{align}
where $\eta_k := \col(x_k - \mu^{\sf x}_k, w_{[k,k+N)}, v_{[k,k+N)}) \in \real^{n_\eta}$
is a vector of uncorrelated zero-mean random variables of dimension $n_\eta := n+nN+pN$, and matrix $\Lambda_t \in \real^{(m+p) \times n_\eta}$ is linearly dependent on the gain matrices $M^s_t$ as
\begin{align} \label{Eq:Lambda_definition}
    \Lambda_t := \compactmat{\Delta^{\sf U}_{t-k} \\[.2em] \Delta^{\sf Y}_{t-k}} \mathcal{M} \, \Delta^{\sf M} + \compactmat{0_{m\times n_\eta} \\[.2em] \Delta^{\sf A}_{t-k}}, \quad t \in \integer_{[k,k+N)},
\end{align}
where $\mathcal{M} \in \real^{mN\times pN}$ is a concatenation of $M^s_t$
\begingroup
\setlength{\abovedisplayskip}{.2em}
\setlength{\belowdisplayskip}{.3em}
\begin{align} \label{Eq:M_concatenation}
    \mathcal{M} := \compactmat{M^k_k \\ M^k_{k+1} & M^{k+1}_{k+1} \\[-.3em] \vdots & \vdots & \ddots \\ M^k_{k+N-1} & M^{k+1}_{k+N-1} & \cdots & M^{k+N-1}_{k+N-1}}
\end{align}
\endgroup
and where $\Delta^{\sf U}_i \in \real^{m \times mN}, \Delta^{\sf Y}_i \in \real^{p \times mN}, \Delta^{\sf A}_i \in \real^{p \times n_\eta}$ and $\Delta^{\sf M} \in \real^{pN \times n_\eta}$ are independent of both decision variables $\thickbar u$ and $M^s_t$, with expressions available in \ref{APPENDIX:Delta_Lambda_definition}.

\subsubsection{Deterministic Formulation of Cost and Constraint}
Given \eqref{Eq:IO_wrt_independent_RV}, $\col(u_t, y_t)$ has mean $\col(\thickbar u_t, \thickbar y_t)$ and variance $\Lambda_t \Sigma^\eta \Lambda^\transpose_t$, since $\eta_k$ has zero mean and the variance $\Sigma^\eta := \Diag(\Sigma^{\sf x}, I_N\otimes\Sigma^{\sf w}, I_N\otimes\Sigma^{\sf v}) \in \symmetric^{n_\eta}_+$ via \eqref{Eq:noise_meanvar} and \eqref{Eq:control_step_state_meanvar}.
Then, the constraint \eqref{Eq:DR_CVaR_Constraint} can be equivalently written as a second-order cone (SOC) constraint of the decision variables $\thickbar u$ and $M^s_t$.

\begin{lemma}[SOC Expression of DR-CVaR Constraint\Versions{ \cite{EXTENDED}}{}] \label{LEMMA:WorstCaseCVaR_SOCP_form}
    With $h(u_t, y_t)$ as in \eqref{Eq:constraint_function}, for $t \in \integer_{[k,k+N)}$, \eqref{Eq:DR_CVaR_Constraint} holds iff
\begin{align} \label{Eq:CVaR_Constraint_SOCP_form}
    2 {\textstyle \big( \frac{1-\alpha} \alpha} \big)\!^{\frac12} \big\Vert (\Sigma^\eta)\!^{\frac12} \Lambda_t^\transpose e_i \big\Vert_2 \leq - e_i^\transpose \neatmat{\thickbar u_t \\[-.25em] \thickbar y_t} + f_i, \;\; i \in \integer_{[1,q]}.
\end{align}
\end{lemma}
\Versions{}{
\begin{proof} \tgrey
\setlength{\abovedisplayskip}{.2em}
\setlength{\belowdisplayskip}{.2em}
    Substituting \eqref{Eq:IO_wrt_independent_RV} into \eqref{Eq:constraint_function}, $h(u_t, y_t)$ can be written as
\begin{align*}
    h(u_t, y_t) = \max_{i \in \{1,\ldots,q\}} e_i^\transpose \Lambda_t \eta_k + e^\transpose_i \col(\thickbar u_t, \thickbar y_t) - f_i,
\end{align*}
    where the random variable $\eta_k$ has zero mean and variance $\Sigma^\eta$.
    According to \cite[Thm. 3.3]{Zymler2013}, \eqref{Eq:DR_CVaR_Constraint} holds if and only if there exist $\theta_t \in \real$ and $\Theta_t \in \symmetric^{n_\eta+1}_+$ satisfying the LMIs
\begin{align*} \begin{aligned}
    0 &\geq \alpha \theta_t + \mathrm{Trace} \big[ \Theta_t \, \Diag(\Sigma^\eta, \, 1)  \big] \\
    \Theta_t &\succeq \compactmat{
        0_{n_\eta \times n_\eta} & \Lambda_t^\transpose e_i \\[.2em] 
        e_i^\transpose \Lambda_t &  e^\transpose_i \col(\thickbar u_t, \thickbar y_t) - f_i - \theta_t
        },\,\, i \in \integer_{[1,q]}.
\end{aligned} \end{align*}
    From \cite[Thm. 1]{Ghaoui2003}, these LMIs are feasible in $(\theta_t,\Theta_t)$ if and only if \eqref{Eq:CVaR_Constraint_SOCP_form} holds, which completes the proof.
\end{proof}}

SMPC problems typically consider the expected cost $\sum_{t=k}^{k+N-1} \mathbb{E}[J_t (u_t, y_t)]$ summing \eqref{Eq:stage_cost} over the horizon, which is equal to a deterministic quadratic function of $\thickbar u$ and $M^s_t$,
\begin{align} \label{Eq:cost_reduced}
    \textstyle{\sum_{t=k}^{k+N-1}} \big[ J_t (\thickbar u_t, \thickbar y_t)
    + \Vert \Diag(R, Q)^{\frac12} \Lambda_t (\Sigma^\eta)^{\frac12} \Vert_{\sf F}^2 \big],
\end{align}
given the mean and variance of $\col(u_t, y_t)$ and given that $\mathbb{E} [ \Vert z \Vert_S^2 ] = \Vert \mathbb{E}[z] \Vert_S^2 + \Vert S^{\frac12} \mathrm{Var}[z]^{\frac12} \Vert_{\sf F}^2$ for any random vector $z$ and fixed matrix $S$; $\Vert \cdot \Vert_{\sf F}$ denotes the Frobenius norm.

\subsubsection{SMPC Optimization Problem and Algorithm}
Using the cost \eqref{Eq:cost_reduced} and reformulation \eqref{Eq:CVaR_Constraint_SOCP_form} of constraint \eqref{Eq:DR_CVaR_Constraint}, we have the SMPC problem as a second-order cone problem (SOCP)
\begin{align} \label{Eq:SMPC_reduced} \begin{aligned}
    \minimize_{\thickbar u, M^s_t} \; \eqref{Eq:cost_reduced} 
    \;\mathrm{s.t.}\; \eqref{Eq:CVaR_Constraint_SOCP_form} \;\text{for}\; t \in \integer_{[k,k+N)}, \eqref{Eq:nominal_state}, \eqref{Eq:Lambda_definition},
\end{aligned} \end{align}
which problem has a unique optimal solution when feasible, since \eqref{Eq:cost_reduced} is jointly strongly convex\Versions{}{{\footnote{\tgrey Strong convexity in $\thickbar u$ is clear from the first term; strong convexity in $M_{t}^{s}$ can be shown by noting that a sub-matrix of $\col(\mathcal{J}_k, \ldots, \mathcal{J}_{k+N-1})$ with $\mathcal{J}_t := \Diag(R,Q)^{1/2} \Lambda_t (\Sigma^\eta)^{1/2}$ is $\bar{\mathcal{J}}_{\sf L} \mathcal{M} \bar{\mathcal{J}}_{\sf R}$, where $\bar{\mathcal{J}}_{\sf L} := I_N \otimes R^{1/2}$ and $\bar{\mathcal{J}}_{\sf R} := (I_{pN} - \Xi(A_{\sf L}) (I_N \otimes L_{\sf L})) (I_N \otimes (\Sigma^{\sf v})^{1/2})$ are non-singular.}}} in $\thickbar u$ and $M^s_t$.

The nominal inputs $\thickbar u$ and gains $M^s_t$ determined from \eqref{Eq:SMPC_reduced} complete the parameterization of control policies $\pi_{[k, k+N)}$ in \eqref{Eq:feedback_policy}, and the upcoming $N_{\rm c}$ control inputs $u_{[k, k+N_{\rm c})}$ are decided by the first $N_{\rm c}$ policies $\pi_{[k, k+N_{\rm c})}$ respectively, with parameter $N_{\rm c} \in \integer_{[1,N]}$.
The next control step will be set as $t = k+N_{\rm c}$, and the state mean $\mu^{\sf x}_{k+N_{\rm c}}$ in \eqref{Eq:control_step_state_meanvar} will be iterated as the estimate $\hat x_{k+N_{\rm c}}$ via \eqref{Eq:state_estimation}; {\tb we let the nominal state $\thickbar x_{k+N_{\rm c}}$ via \eqref{Eq:nominal_state} be a backup value $\mu^{\sf \bar x}_{k+N_{\rm c}}$ of $\mu^{\sf x}_{k+N_{\rm c}}$ that ensures feasibility of \eqref{Eq:SMPC_reduced} at the new control step \cite{OFSMPC:farina2015}}.
The entire SMPC control process is shown in Algorithm \ref{ALGO:SMPC}.
\begin{algorithm}
\caption{Distributionally Robust Optimized-Gain Stochastic MPC (DR/O-SMPC)} \label{ALGO:SMPC}
\begin{algorithmic}[1]
    \Require horizon lengths $N, N_{\rm c}$, system matrices $A,B,C$, noise variances $\Sigma^{\sf w}, \Sigma^{\sf v}$, initial-state mean $\mu^{\sf x}_\initial$, cost matrices $Q, R$, constraint coefficients $E, f$, and CVaR level $\alpha$.
    \State Compute $\Sigma^{\sf x}, L_{\sf L}$ via \eqref{Eq:Kalman_gain} and $\Delta^{\sf U}_{[0,N)}, \Delta^{\sf Y}_{[0,N)}, \Delta^{\sf A}_{[0,N)}, \Delta^{\sf M}$ through \ref{APPENDIX:Delta_Lambda_definition}.
    \State Initialize the control step $k \gets 0$ and set $\mu^{\sf x}_0 \gets \mu^{\sf x}_\initial$.
    \State Solve $\thickbar u_{[k,k+N)}$ and $M^s_t$ from problem \eqref{Eq:SMPC_reduced}. \label{LINE:SMPC:solving}
    \State {\tb {\bf If} \eqref{Eq:SMPC_reduced} is infeasible {\bf then} Set $\mu^{\sf x}_k \gets \mu^{\sf \bar x}_k$, and redo line \ref{LINE:SMPC:solving}.}
    \For{{\bf $t$ from $k$ to $k+N_{\rm c}-1$}}
        \State Input $u_t \gets \pi_t ( \nu_{[k,t)} )$ in \eqref{Eq:feedback_policy} to the system \eqref{Eq:LTI}.
        \State Measure $y_t$ from the system \eqref{Eq:LTI}.
        \State Compute $\nu_t$ via \eqref{Eq:state_estimation}.
    \EndFor
    \State Set $(\mu^{\sf x}_{k+N_{\rm c}}, \mu^{\sf \bar x}_{k+N_{\rm c}})$ as $(\hat x_{k+N_{\rm c}}, \thickbar x_{k+N_{\rm c}})$ in \eqref{Eq:state_estimation}, \eqref{Eq:nominal_state}.
    \State Set $k \gets k + N_{\rm c}$. Go back to line \ref{LINE:SMPC:solving}.
\end{algorithmic}
\end{algorithm}

\subsection{Stochastic Data-Driven Predictive Control (SDDPC)} \label{SECTION:theory:SDDPC}

We develop in this section a data-driven control method, which consists of an offline process for data collection and an online process that makes real-time control decisions.


\subsubsection{Use of Offline Data}
In data-driven control, sufficient offline data is required to capture the system's behavior. Here we explain how we collect data and use it to calculate some quantities required in our control method. We first consider noise-free data and then address the case of noisy data.

Consider a deterministic version of the system \eqref{Eq:LTI}
\begin{equation} \label{Eq:LTI_deterministic}
    x_{t+1} = A x_t + B u_t, \qquad y_t = C x_t + D u_t.
\end{equation}
By assumption, \eqref{Eq:LTI_deterministic} is minimal; let $L \in \natural$ be such that the extended observability matrix $\mathcal{O} := \col(C, CA, \ldots, CA^{L-1})$ has full column rank.
Let $u^\data_{[1,T_\data]}, y^\data_{[1,T_\data]}$ be a $T_\data$-length trajectory of input-output data collected from \eqref{Eq:LTI_deterministic}. The input sequence $u^\data$ is assumed to be \emph{persistently exciting} of order $K_\data := L + 1 + n$, i.e., its associated $K_\data$-depth block-Hankel matrix $\Hank_{K_\data} ( u^\data_{[1,T_\data]} ) \in \real^{m K_\data \times (T_\data - K_\data + 1)}$, defined as 
\begingroup
\setlength{\abovedisplayskip}{.1em}
\setlength{\belowdisplayskip}{.2em}
\begin{align*}
    \Hank_{K_\data}(u^\data_{[1,T_\data]}) := \compactmat{
        u^\data_1 & u^\data_2 & \cdots & u^\data_{T_\data-K_\data+1} \\
        u^\data_2 & u^\data_3 & \cdots & u^\data_{T_\data-K_\data+2} \\[-.5em]
        \scalebox{0.8}{\bf$\vdots$} & \scalebox{0.8}{\bf$\vdots$} & \scalebox{0.8}{\bf$\ddots$} & \scalebox{0.8}{\bf$\vdots$} \\
        u^\data_{K_\data} & u^\data_{K_\data+1} & \cdots & u^\data_{T_\data} }, 
\end{align*}
\endgroup
has full row rank. We formulate data matrices $U_1 \in \real^{mL \times h}$, $U_2 \in \real^{m \times h}$, $Y_1 \in \real^{pL \times h}$ and $Y_2 \in \real^{p \times h}$ of width $h := T_{\rm d} - L$ by partitioning associated Hankel matrices as
\begin{align} \label{Eq:data_matrices} \begin{aligned}
    \compactmat{U_1\\ U_2} := \Hank_{L+1} ( u^\data_{[1,T_\data]} ), \quad
    \compactmat{Y_1\\ Y_2} := \Hank_{L+1} ( y^\data_{[1,T_\data]} ).
\end{aligned} \end{align}
The data matrices in \eqref{Eq:data_matrices} will now be used to represent a quantity $\mathbf{\Gamma} \in \real^{p \times (mL+pL)}$ related to the system \eqref{Eq:LTI_deterministic}, 
\begin{align}
\label{Eq:Gamma_definition}
    \mathbf{\Gamma} &= \mat{\mathbf{\Gamma}_{\sf U} & \mathbf{\Gamma}_{\sf Y}} := \mat{C \mathcal{C} & C A^L} \compactmat{I_{mL} \\ \mathcal{G} & \mathcal{O}}^\dagger,
\end{align}
with $\mathcal{C} := [A^{L-1}B, \ldots, AB, B]$ the extended controllability matrix and $\mathcal{G} := \Toep(D, CB, \ldots, C A^{L-2}B)$ the impulse-response matrix; $\Toep$ denotes the block-Toeplitz matrix
\begingroup
\setlength{\abovedisplayskip}{.3em}
\setlength{\belowdisplayskip}{.5em}
\begin{align*} \Toep(M_1, \ldots, M_k) := \smallmat{M_1 \\ M_2 & M_1 \\[-.4em] \scalebox{0.6}{\bf$\vdots$} & \scalebox{0.6}{\bf$\ddots$} & \scalebox{0.6}{\bf$\ddots$} \\ M_k & \cdots & M_2 & M_1}. \end{align*}
\endgroup

\begin{lemma} [Data Representation of $\mathbf{\Gamma}$ and $D$ \cite{SDDPC}] \label{LEMMA:system_quantity_data_representation}
    If system \eqref{Eq:LTI_deterministic} is controllable and the input data $u^\data_{[1,T_\data]}$ is persistently exciting of order $L+1+n$, then, given the data matrices in \eqref{Eq:data_matrices}, the matrix $\mathbf{\Gamma}$ defined in \eqref{Eq:Gamma_definition} and matrix $D$ in system \eqref{Eq:LTI_deterministic} can be expressed as $[\mathbf{\Gamma}_{\sf U}, \mathbf{\Gamma}_{\sf Y}, D] = Y_2 \, \col(U_1, Y_1, U_2)^\dagger$.
\end{lemma}

With Lemma \ref{LEMMA:system_quantity_data_representation}, the matrices $\mathbf{\Gamma}, D$ are represented using offline data collected from system \eqref{Eq:LTI_deterministic}, and will be used as part of the construction for our data-driven control method.

In the case where the measured data is corrupted by noise, as will usually be the case, the pseudo-inverse computation in Lemma \ref{LEMMA:system_quantity_data_representation} is numerically fragile and does not recover the desired matrices $\mathbf{\Gamma}, D$. A standard technique to robustify this computation is to replace the pseudo-inverse $W^\dagger$ of $W := \col(U_1, Y_1, U_2)$ in Lemma \ref{LEMMA:system_quantity_data_representation} with its Tikhonov regularization $(W^\transpose W + \lambda I_h)^{-1} W^\transpose$ with a regularization parameter $\lambda > 0$. 


\subsubsection{Auxiliary State-Space Model}
The SMPC approach of Section \ref{SECTION:theory:SMPC} uses as sub-components a state estimator, an affine feedback law and a DR-CVaR constraint.
We now leverage the offline data as described in Section \ref{SECTION:theory:SDDPC}-1 to directly design analogs of these components based on data, without knowledge of the system matrices. 

We begin by constructing an auxiliary state-space model which has equivalent input-output behavior to \eqref{Eq:LTI}, but is parameterized only by the recorded data sequences.
Define auxiliary signals $\mathbf{x}_t, \mathbf{w}_t \in \real^{n_\auxiliary}$ of dimension $n_\auxiliary := mL+pL+pL^2$ for system \eqref{Eq:LTI} by
\begin{align} \label{Eq:DDModel_state_definition}
    \mathbf{x}_t := \left[ \begin{array}{c} u_{[t-L,t)} \\ \hline y^\circ_{[t-L,t)} \\ \hline \rho_{[t-L,t)} \end{array} \right], \quad
    \mathbf{w}_t := \left[ \compact{\begin{array}{c} 0_{mL\times 1} \\ \hline 0_{pL\times 1} \\ \hline 0_{pL(L-1)\times 1} \\ \rho_t \end{array}} \right]
\end{align}
where $y^\circ_t := y_t - v_t \in \real^p$ is the output excluding measurement noise, and $\rho_t := \mathcal{O} w_t \in \real^{pL}$ stacks the system's response to process noise $w_t$ on time interval $[t+1, t+L]$.
The auxiliary signals $\mathbf{x}_t, \mathbf{w}_t$ together with $u_t, y_t, v_t$ then satisfy the relations given by Lemma \ref{LEMMA:AuxModel}. 
\begin{lemma}[Auxiliary Model {\cite{SDDPC}}] \label{LEMMA:AuxModel}
    For system \eqref{Eq:LTI}, signals $u_t, y_t, v_t$ and the auxiliary signals $\mathbf{x}_t, \mathbf{w}_t$ in \eqref{Eq:DDModel_state_definition} satisfy
\begin{align} \label{Eq:DDModel}
    \mathbf{x}_{t+1} = \mathbf{A} \mathbf{x}_t + \mathbf{B} u_t + \mathbf{w}_t, \quad
    y_t = \mathbf{C} \mathbf{x}_t + D u_t + v_t
\end{align}
    with $\mathbf{A} \in \real^{n_\auxiliary \times n_\auxiliary}$, $\mathbf{B} \in \real^{n_\auxiliary \times m}$, $\mathbf{C} \in \real^{p \times n_\auxiliary}$ given by
\begin{align*} \begin{aligned}
    \mathbf{A} &:= \col(0_{mL \times n_\auxiliary}, 0_{p(L-1) \times n_\auxiliary}, \mathbf{C}, 0_{pL^2 \times n_\auxiliary}) \\
    &\quad\; + \Diag(\mathcal{D}_m, \mathcal{D}_p, \mathcal{D}_{pL}), \;\text{with}\; \mathcal{D}_q := \smallmat{& I_{q(L-1)} \\[-.25em] 0_{q\times q}}, \\
    \mathbf{B} &:= \col \big( 0_{m(L-1) \times m}, I_m, 0_{p(L-1)\times m}, D, 0_{pL^2\times m}), \\
    \mathbf{C} &:= \left[ \mathbf{\Gamma}_{\sf U}, \mathbf{\Gamma}_{\sf Y}, \mathbf{F} - \mathbf{\Gamma}_{\sf Y} \mathbf{E} \right],
\end{aligned} \end{align*}
with matrices $\mathbf{\Gamma}_{\sf U}, \mathbf{\Gamma}_{\sf Y}$ in \eqref{Eq:Gamma_definition}, and zero-one matrices $\mathbf{E} := \Toep(0_{p\times pL}, S_1, \ldots, S_{L-1})$ and $\mathbf{F} := [S_L, S_{L-1}, \ldots, S_1]$ composed by $S_j := [ 0_{p\times (j-1)p} , I_p , 0_{p\times (L-j)p} ]$ for $j \in \integer_{[1,L]}$.
\end{lemma}
The output noise $v_t$ in \eqref{Eq:DDModel} is precisely the same as in \eqref{Eq:LTI}; $\mathbf{w}_t$ appears now as a new disturbance of zero mean and the variance $\mathbf{\Sigma}^{\sf w} := \Diag ( 0_{(n_\auxiliary - pL) \times (n_\auxiliary - pL)}, \Sigma^\rho)$,
where $\Sigma^\rho := \mathcal{O} \Sigma^{\sf w} \mathcal{O}^\transpose \in \symmetric^{pL}_+$ is the variance of $\rho_t$.
The matrices $\mathbf{A}, \mathbf{B}, \mathbf{C}, D$ in \eqref{Eq:DDModel} are known given offline data described in Section \ref{SECTION:theory:SDDPC}-1, since they only depend on $\mathbf{\Gamma}_{\sf U}, \mathbf{\Gamma}_{\sf Y}, D$ which are data-representable via Lemma \ref{LEMMA:system_quantity_data_representation}.
Hence, the auxiliary model \eqref{Eq:DDModel} can be interpreted as a data-representable realization of the system \eqref{Eq:LTI}. 


\subsubsection{Data-Driven State Estimation, Feedback and Constraint}
The auxiliary model \eqref{Eq:DDModel} will now be used for both state estimation and constrained feedback control purposes. Suppose we are at a control step $t=k$ in a receding-horizon process.
Similar to \eqref{Eq:control_step_state_meanvar}, auxiliary state $\mathbf{x}_k$ has condition $\mathbb{E}[\mathbf{x}_k] = \boldsymbol{\mu}^{\sf x}_k$ and $\mathrm{Var}[\mathbf{x}_k] = \mathbf{\Sigma}^{\sf x}$,
where $\boldsymbol{\mu}^{\sf x}_k$ is known from the state estimator to be introduced next; at the initial time $k=0$, the initial mean $\boldsymbol{\mu}^{\sf x}_\initial$ is a parameter;
{\tb the variance $\mathbf{\Sigma}^{\sf x}$ is the unique positive semi-definite solution to DARE \eqref{Eq:DD:Kalman_gain:state_variance},
\begin{subequations} \label{Eq:DD:Kalman_gain} \begin{align}
\label{Eq:DD:Kalman_gain:state_variance}
    & \mathbf{\Sigma}^{\sf x} = (\mathbf{A} - \mathbf{L}_{\sf L} \mathbf{C}) \mathbf{\Sigma}^{\sf x} \mathbf{A}^\transpose + \mathbf{\Sigma}^{\sf w} \\
\label{Eq:DD:Kalman_gain:gains}
    & \mathbf{L}_{\sf L} := \mathbf{A} \mathbf{\Sigma}^{\sf x} \mathbf{C}^\transpose (\mathbf{C} \mathbf{\Sigma}^{\sf x} \mathbf{C}^\transpose + \Sigma^{\sf v})^{-1}
\end{align} \end{subequations}
given $(\mathbf{A}, \mathbf{C})$ detectable and $(\mathbf{A}, \mathbf{\Sigma}^{\sf w})$ stabilizable \cite[Lemma 5]{SDDPC}.}
The state estimator for the auxiliary model \eqref{Eq:DDModel} is analogous to \eqref{Eq:state_estimation}, with observer gain $\mathbf{L}_{\sf L} \in \real^{n_\auxiliary \times p}$ in \eqref{Eq:DD:Kalman_gain:gains}, 
\begin{subequations} \label{Eq:DD:state_estimation} \begin{align}
    \label{Eq:DD:state_estimation:innovation}
    \boldsymbol{\nu}_t &:= y_t - \mathbf{C} \hat {\mathbf{x}}_t - D u_t, &&  t \in \integer_{[k,k+N)} \\
    \label{Eq:DD:state_estimation:iteration}
    \hat{\mathbf{x}}_{t+1} &:= \mathbf{A} \hat{\mathbf{x}}_t + \mathbf{B} u_t + \mathbf{L}_{\sf L} \boldsymbol{\nu}_t, &&  t \in \integer_{[k,k+N)}
    \\
    \label{Eq:DD:state_estimation:initial}
    \hat{\mathbf{x}}_k &:= \boldsymbol{\mu}^{\sf x}_k
\end{align} \end{subequations}
where $\hat{\mathbf{x}}_t$ is the estimate and $\boldsymbol{\nu}_t$ is the innovation.
The output-error-feedback policy \eqref{Eq:feedback_policy} in SMPC is now extended as
$\boldsymbol{\pi}_t(\cdot)$,
\begin{align} \label{Eq:DD:feedback_policy}
    u_t \gets \boldsymbol{\pi}_t ( \boldsymbol{\nu}_{[k,t)} ) := \thickbar u_t + \textstyle{\sum_{s=k}^{t-1}} \, M^s_t \, \boldsymbol{\nu}_s
\end{align}
where the nominal input $\thickbar u_t \in \real^m$ and gain matrices $M^s_t \in \real^{m\times p}$ are decision variables.
Let $\thickbar{\mathbf{x}}_t \in \real^{n_\auxiliary}$ and $\thickbar{\mathbf{y}}_t \in \real^p$ be the resulting nominal state and nominal output as
\begin{subequations} \label{Eq:DD:nominal_state} \begin{align}
    \label{Eq:DD:nominal_state:recursion}
    \thickbar{\mathbf{x}}_{t+1} &:= \mathbf{A} \thickbar{\mathbf{x}}_t + \mathbf{B} \thickbar u_t, & t \in \integer_{[k,k+N)}, \\
    \label{Eq:DD:nominal_state:output}
    \thickbar{\mathbf{y}}_t &:= \mathbf{C} \thickbar{\mathbf{x}}_t + D \thickbar u_t, & t \in \integer_{[k,k+N)}, \\
    \label{Eq:DD:nominal_state:initial}
    \thickbar{\mathbf{x}}_k &:= \boldsymbol{\mu}^{\sf x}_k.
\end{align} \end{subequations}
The SOC formulation of constraint \eqref{Eq:DR_CVaR_Constraint} is similar to \eqref{Eq:CVaR_Constraint_SOCP_form},
\begin{align} \label{Eq:DD:CVaR_Constraint_SOCP_form} \begin{aligned}
    2 {\textstyle \big(\frac{1-\alpha} \alpha} \big)\!^{\frac12} \big\Vert (\mathbf{\Sigma}^\eta)\!^{\frac12} \mathbf{\Lambda}_t^\transpose e_i \big\Vert_2 \leq - e_i^\transpose \neatmat{\thickbar u_t \\[-.25em] \thickbar{\mathbf{y}}_t} + f_i, \; i \in \integer_{[1,q]}
\end{aligned} \end{align}
with matrices $\mathbf{\Sigma}^\eta := \Diag(\mathbf{\Sigma}^{\sf x}, I_N \otimes \mathbf{\Sigma}^{\sf w}, I_N \otimes \Sigma^{\sf v}) \in \symmetric^{n_{\eta \dash \auxiliary}}_+$ and $\mathbf{\Lambda}_t \in \real^{(m+p) \times n_{\eta \dash \auxiliary}}$ with $n_{\eta \dash \auxiliary} := n_\auxiliary + n_\auxiliary N + p N$,
\begin{align} \label{Eq:DD:Lambda_definition}
    \mathbf{\Lambda}_t := \compactmat{\mathbf{\Delta}^{\sf U}_{t-k} \\[.2em] \mathbf{\Delta}^{\sf Y}_{t-k}} \mathcal{M} \mathbf{\Delta}^{\sf M} + \compactmat{0_{m \times n_{\eta \dash \auxiliary}} \\[.2em] \mathbf{\Delta}^{\sf A}_{t-k}}
\end{align}
where $\mathbf{\Delta}^{\sf U}_i \in \real^{m \times mN}$, $\mathbf{\Delta}^{\sf Y}_i \in \real^{p \times mN}$, $\mathbf{\Delta}^{\sf A}_i \in \real^{p \times n_{\eta \dash \auxiliary}}$ and $\mathbf{\Delta}^{\sf M} \in \real^{pN \times n_{\eta \dash \auxiliary}}$ can be found in \ref{APPENDIX:Delta_Lambda_definition}, and where $\mathcal{M} \in \real^{mN \times pN}$ is a concatenation of $M^s_t$ as in \eqref{Eq:M_concatenation}.


\subsubsection{SDDPC Optimization Problem and Algorithm}
With the results above, we are now ready to mirror the steps of getting \eqref{Eq:SMPC_reduced}
and formulate a distributionally robust optimized-gain Stochastic Data-Driven Predictive Control (SDDPC) problem,
\begin{align} \label{Eq:SDDPC_reduced} \begin{aligned}
    \minimize_{\thickbar u, M^s_t} \; \eqref{Eq:DD:cost_reduced}
    \;\mathrm{s.t.}\; \eqref{Eq:DD:CVaR_Constraint_SOCP_form} \;\text{for}\; t \in \integer_{[k,k+N)}, \eqref{Eq:DD:nominal_state}, \eqref{Eq:DD:Lambda_definition}
\end{aligned} \end{align}
where the quadratic cost function is analogous to \eqref{Eq:cost_reduced} as
\begin{align} \label{Eq:DD:cost_reduced}
    \textstyle{\sum_{t=k}^{k+N-1}} \big[ J_t (\thickbar u_t, \thickbar{\mathbf{y}}_t)
    + \Vert \Diag(R, Q)^{\frac12} \mathbf{\Lambda}_t (\mathbf{\Sigma}^\eta)^{\frac12} \Vert_{\sf F}^2 \big].
\end{align}
Problem \eqref{Eq:SDDPC_reduced} has a unique optimal solution if feasible, similar as problem \eqref{Eq:SMPC_reduced}. The solution $(\thickbar u, M^s_t)$ finishes parameterization of the control policies $\boldsymbol{\pi}_{[k, k+N)}$ via \eqref{Eq:DD:feedback_policy}, where the first $N_{\rm c}$ policies are applied to the system.
At the next control step $t = k + N_{\rm c}$, the state mean $\boldsymbol{\mu}^{\sf x}_{k+N_{\rm c}}$ is iterated as the estimate $\hat{\mathbf{x}}_{k+N_{\rm c}}$ via \eqref{Eq:state_estimation}, {\tb with a backup value $\boldsymbol{\mu}^{\sf \bar x}_{k+N_{\rm c}}$ of $\boldsymbol{\mu}^{\sf x}_{k+N_{\rm c}}$ equal to the nominal state $\thickbar{\mathbf{x}}_{k+N_{\rm c}}$ via \eqref{Eq:nominal_state}.} 
The method is formally summarized in Algorithm \ref{ALGO:SDDPC}.

\begin{algorithm}
\caption{Distributionally Robust Optimized-Gain Stochastic Data-Driven Predictive Control (DR/O-SDDPC)} \label{ALGO:SDDPC}
\begin{algorithmic}[1]
    \Require horizon lengths $L, N, N_{\rm c}$, offline data $u^\data, y^\data$, noise variances $\Sigma^\rho, \Sigma^{\sf v}$, initial-state mean $\boldsymbol{\mu}^{\sf x}_\initial$, cost matrices $Q, R$, constraint coefficients $E, f$, and CVaR level $\alpha$.
    \State Compute $\mathbf{\Gamma}$ and $D$ as in Section \ref{SECTION:theory:SDDPC}-1 using data $u^\data, y^\data$, and formulate matrices $\mathbf{A}, \mathbf{B}, \mathbf{C}$ as in Section \ref{SECTION:theory:SDDPC}-2.
    \State Compute $\mathbf{\Sigma}^{\sf x}, \mathbf{L}_{\sf L}$ via \eqref{Eq:DD:Kalman_gain} and $\mathbf{\Delta}^{\sf U}_{[0,N)}, \mathbf{\Delta}^{\sf Y}_{[0,N)},  \mathbf{\Delta}^{\sf A}_{[0,N)}$, $\mathbf{\Delta}^{\sf M}$ through \ref{APPENDIX:Delta_Lambda_definition}. 
    \State Initialize the control step $k \gets 0$ and set $\boldsymbol{\mu}^{\sf x}_0 \gets \boldsymbol{\mu}^{\sf x}_\initial$.
    \State Solve $\thickbar u_{[k,k+N)}$ and $M^s_t$ from problem \eqref{Eq:SDDPC_reduced}. \label{LINE:SDDPC:solving}
    \State {\tb {\bf If} \eqref{Eq:SDDPC_reduced} is infeasible {\bf then} Set $\boldsymbol{\mu}^{\sf x}_k \gets \boldsymbol{\mu}^{\sf \bar x}_k$, and redo line \ref{LINE:SDDPC:solving}.}
    \For{{\bf $t$ from $k$ to $k+N_{\rm c}-1$}}
        \State Input $u_t \gets \boldsymbol{\pi}_t ( \boldsymbol{\nu}_{[k,t)} )$ in \eqref{Eq:DD:feedback_policy} to the system \eqref{Eq:LTI}.
        \State Measure $y_t$ from the system \eqref{Eq:LTI}.
        \State Compute $\boldsymbol{\nu}_t$ via \eqref{Eq:DD:state_estimation}.
    \EndFor
    \State Set $(\boldsymbol{\mu}^{\sf x}_{k+N_{\rm c}}, \boldsymbol{\mu}^{\sf \bar x}_{k+N_{\rm c}})$ as $(\hat{\mathbf{x}}_{k+N_{\rm c}}, \thickbar{\mathbf{x}}_{k+N_{\rm c}})$ in \eqref{Eq:DD:state_estimation}, \eqref{Eq:DD:nominal_state}
    \State Set $k \gets k + N_{\rm c}$. Go back to line \ref{LINE:SDDPC:solving}
\end{algorithmic}
\end{algorithm}

\subsection{Theoretical Equivalence of SMPC and SDDPC} \label{SECTION:theory:equivalence}


We establish theoretical results in this section, starting by an underlying relation between the means of $x_k$ and $\mathbf{x}_k$. 

\begin{lemma}[Related Means of $x_k$ and $\mathbf{x}_k$ \cite{SDDPC}] 
    If $\mu^{\sf x}_k$ is the mean of $x_k$ and $\boldsymbol{\mu}^{\sf x}_k$ is the mean of $\mathbf{x}_k$, then they satisfy
\begin{align} \label{Eq:LEMMA:state_distribution_relation}
    \mu^{\sf x}_k = \Phi_\original \,\widetilde{\mu}^{\sf \, x}_k, \qquad
    \boldsymbol{\mu}^{\sf x}_k = \Phi_\auxiliary \,\widetilde{\mu}^{\sf \, x}_k
\end{align}
    for some $\widetilde{\mu}^{\sf \, x}_k \in \real^{mL+n(L+1)}$, where matrices $\Phi_\original, \Phi_\auxiliary$ are
\begin{align*}
    \Phi_\original := \big[ \mathcal{C} ,\, A^L ,\, \mathcal{C}_{\sf w} \big], \quad
    \Phi_\auxiliary := \smallmat{ I_{mL} && \\ \mathcal{G} & \;\;\mathcal{O}\; & \mathcal{G}_{\sf w} \\ && I_L \otimes \mathcal{O} },
\end{align*}
    with the matrices $\mathcal{C}$, $\mathcal{O}$, $\mathcal{G}$ defined in Section \ref{SECTION:theory:SDDPC}-1 and $\mathcal{C}_{\sf w} := [A^{L-1},\ldots,A,I_n]$, $\mathcal{G}_{\sf w} := \Toep(0_{p\times n}, C, CA, \ldots, CA^{L-2})$.
\end{lemma}

As we assume \eqref{Eq:LEMMA:state_distribution_relation} holds, the SMPC and SDDPC problems will have equal feasible and optimal sets.

\begin{proposition}[Equivalence of Optimization Problems\Versions{ \cite{EXTENDED}}{}] \label{PROPOSITION:equivalence_of_optimization_problems}
    If the parameters $\mu^{\sf x}_k, \boldsymbol{\mu}^{\sf x}_k$ satisfy \eqref{Eq:LEMMA:state_distribution_relation},
    then the optimal (resp. feasible) solution set of SDDPC problem \eqref{Eq:SDDPC_reduced} is equal to the optimal (resp. feasible) solution set of SMPC problem \eqref{Eq:SMPC_reduced}.
\end{proposition}
\Versions{}{
\begin{proof} \tgrey
    We first claim that, for all $\thickbar u$ and $M^s_t$, we have
\begin{align} \label{Eq:PROOF:equivalence:target}
    \thickbar y_t = \thickbar{\mathbf{y}}_t, \qquad
    \Lambda_t \Sigma^\eta \Lambda^\transpose_t = \mathbf{\Lambda}_t \mathbf{\Sigma}^\eta \mathbf{\Lambda}^\transpose_t
\end{align}
    for $t \in \integer_{[k,k+N)}$, which is explained in \ref{APPENDIX:proof}. 
    Given \eqref{Eq:PROOF:equivalence:target}, the objective function \eqref{Eq:cost_reduced} of problem \eqref{Eq:SMPC_reduced} and objective function \eqref{Eq:DD:cost_reduced} of problem \eqref{Eq:SDDPC_reduced} are equal, and the constraint \eqref{Eq:CVaR_Constraint_SOCP_form} in problem \eqref{Eq:SMPC_reduced} and constraint \eqref{Eq:DD:CVaR_Constraint_SOCP_form} in problem \eqref{Eq:SDDPC_reduced} are equivalent. 
    Thus the problems \eqref{Eq:SMPC_reduced} and \eqref{Eq:SDDPC_reduced} have the same objective function and constraints, and the result follows.
\end{proof} }

We present in Theorem \ref{PROPOSITION:equivalence_of_control_algorithms} our main theoretical result, saying that our proposed SDDPC control method and the benchmark SMPC method will result in identical control actions, under idealized conditions in Assumption \ref{ASSUMPTION:parameter_choice_of_DDSMPC_algorithm}.

\begin{assumption}[SDDPC Parameter Choice w.r.t. SMPC] \label{ASSUMPTION:parameter_choice_of_DDSMPC_algorithm}
    Given the parameters in Algorithm \ref{ALGO:SMPC}, we assume the parameters in Algorithm \ref{ALGO:SDDPC} satisfy the following.
\begin{enumerate}[(a)]
    \item $L$ is sufficiently large so that $\mathcal{O}$ has full column rank.
    \item Data $u^\data, y^\data$ comes from the deterministic system \eqref{Eq:LTI_deterministic}; the input data $u^\data$ is persistently exciting of order $L+1+n$.
    \item Given $\Sigma^{\sf w}$ in Algorithm \ref{ALGO:SMPC}, the parameter $\Sigma^\rho$ in Algorithm \ref{ALGO:SDDPC} is set equal to $\mathcal{O} \Sigma^{\sf w} \mathcal{O}^\transpose$.
    \item Given $\mu^{\sf x}_\initial$ in Algorithm \ref{ALGO:SMPC}, the parameter $\boldsymbol{\mu}^{\sf x}_\initial$ in Algorithm \ref{ALGO:SDDPC} is selected as $\Phi_\auxiliary \widetilde{\mu}^{\sf \,x}_\initial$ for some $\widetilde{\mu}^{\sf \,x}_\initial \in \real^{mL+n+nL}$ satisfying $\mu^{\sf x}_\initial = \Phi_\original \widetilde{\mu}^{\sf \,x}_\initial$. (Such $\widetilde{\mu}^{\sf \,x}_\initial$ always exists because $\Phi_\original$ has full row rank.)
\end{enumerate}
\end{assumption}

\begin{theorem}[\bf Equivalence of SMPC and SDDPC] \label{PROPOSITION:equivalence_of_control_algorithms}
    Consider system \eqref{Eq:LTI} with initial state $x_0$ and a specific noise realization $\{w_t, v_t\}_{t=0}^\infty$, and consider the following two processes: 
\begin{enumerate}[a)]
    \item decide control actions $\{u_t\}_{t=0}^\infty$ by executing Algorithm \ref{ALGO:SMPC};
    \item decide control actions $\{u_t\}_{t=0}^\infty$ by executing Algorithm \ref{ALGO:SDDPC}, where the parameters satisfy Assumption \ref{ASSUMPTION:parameter_choice_of_DDSMPC_algorithm}.
\end{enumerate}
    Then, the state-input-output trajectories $\{x_t, u_t, y_t\}_{t=0}^\infty$ resulting from process a) and from process b) are the same.
\end{theorem}
\begin{proof}
    The proof is similar to the proof of \cite[Thm. 9]{SDDPC}, requiring Proposition \ref{PROPOSITION:equivalence_of_optimization_problems} and the fact that both problems \eqref{Eq:SMPC_reduced} and \eqref{Eq:SDDPC_reduced} have unique optimal solutions if feasible.
\end{proof}

While in practice Assumption \ref{ASSUMPTION:parameter_choice_of_DDSMPC_algorithm} may not hold, noisy offline data can be accommodated as discussed in Section \ref{SECTION:theory:SDDPC}-1, and $\Sigma^{\rho}$ becomes a tuning parameter of our SDDPC method.
\section{Numerical Case Study} \label{SECTION:simulations}

In this section, we numerically test our proposed method on a batch reactor system \Versions{applied in e.g. \cite{PCE:pan2022b}}{introduced in \cite{SimuModel:Walsh2001} and applied in \cite{DePersis2019, PCE:pan2022b}}.
The system has $n=4$ states, $m=2$ inputs and $p=2$ outputs, and the discrete-time system matrices with sampling period $0.1$s are
\begingroup
\setlength{\abovedisplayskip}{.2em}
\setlength{\belowdisplayskip}{.2em}
\begin{align*}
    \left[ \begin{array}{c|c} A & B \\ \hline C & \end{array} \right] = \left[ \compact{ \begin{array}{cccc|cc} 
    1.178 & .001 & .511 & \wideminus.403 & .004 & \wideminus.087 \\
    \wideminus.051 & .661 & \wideminus.011 & .061 & .467 & .001 \\
    .076 & .335 & .560 & .382 & .213 & \wideminus.235 \\
    0 & .335 & .089 & .849 & .213 & \wideminus.016 \\ \hline
    1 & 0 & 1 & \wideminus1 & & \\ 0 & 1 & 0 & 0 & & \end{array} } \right].
\end{align*} 
\endgroup
The process/sensor noise on each state/output follows the $t$-distribution of 2 DOFs scaled by $10^{-4}$, which is a heavy-tailed distribution. Control parameters are reported in TABLE \ref{TABLE:controller_parameters}.
We collected offline data of length $T_\data = 600$ from the noisy system, where the input data was the outcome of a PI controller $U(s) = \smallmat{0 & -1/s \\ 2+1/s & 0} Y(s)$ plus a white-noise signal of noise power $10^{-2}$.
In the online control process, the reference signal is $r_t = [0,0]^\transpose$ from time $0$s to time $30$s, alternates between $[0,0]^\transpose$ and $[0.3,0]^\transpose$ from $30$s to $60$s, and is $r_t = [0.5,0]^\transpose$ from $60$s to $90$s.
With our proposed SDDPC method, the first output signal is in Fig. \ref{FIG:output}; the signal remains around $0.4$ from $60$s to $90$s because of the safety constraint specified in TABLE \ref{TABLE:controller_parameters}.

For comparison purposes, we implemented the simulation with different controllers.
In addition to distributionally robust optimized-gain (DR/O) SMPC and SDDPC in this paper, we applied the SMPC and SDDPC frameworks from \cite{SDDPC}, which use chance constraints and a fixed feedback gain (CC/F). To observe separate impacts of using the DR constraint and optimized gains, we also implement SMPC and SDDPC with DR constraints and a fixed feedback gain (DR/F).
We also compare to DeePC, SPC and deterministic MPC as benchmarks. The model used in MPC methods is identified from the same offline data in the data-driven controllers.

The simulation results are summarized in TABLE \ref{TABLE:simulation_results}. 
We evaluate (i) the controllers' tracking performance through the tracking cost from $0$s to $60$s and (ii) the controllers' ability to satisfy constraints according to the cumulative amount of constraint violation between $60$s and $90$s, when the first output signal hits the constraint margin.
When the reference signal is constant ($0$s--$30$s), SMPC and SDDPC tracked better than other methods, aligning with the observation in \cite{SDDPC}.
Comparing DR/F and CC/F methods, the controllers with DR constraints achieved lower amounts of constraint violation ($60$s--$90$s), while the tracking performance is slightly worse during $30$s--$60$s when the reference signal has frequent step changes.
Comparing DR/O and DR/F methods, we observe that the methods with optimized gain achieved lower tracking costs when the reference signal changes frequently ($30$s--$60$s).


\begin{figure}[t]
\includegraphics[trim={1cm 0cm 1cm .5cm},clip,width=.48\textwidth]{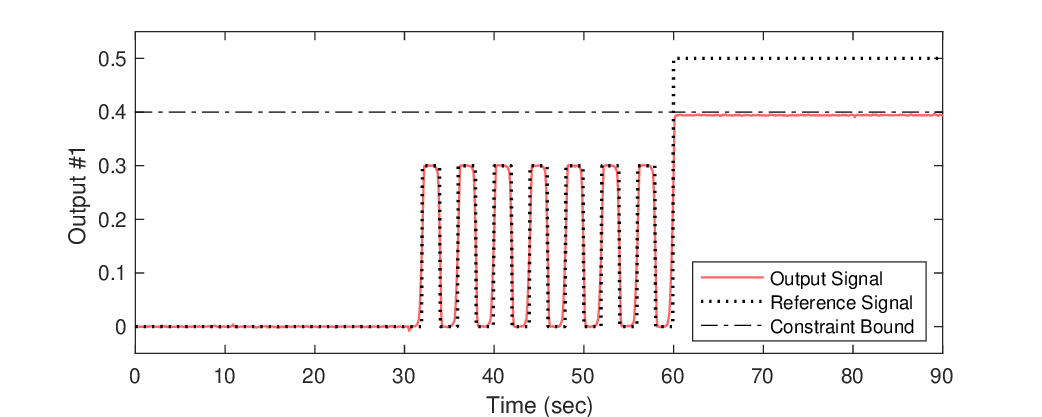}\vspace{-.2em}
\caption{The system's first output signal with DR/O-SDDPC.}
\label{FIG:output}
\end{figure}

{\afterpage{
\begin{table}[t] 
\caption{Control Parameters}
\label{TABLE:controller_parameters} 
\centering 
\begin{tabular}{@{}ll@{}} 
    \toprule
    Time horizon lengths & $L = 5$, $N = 15$, $N_{\rm c} = 5$ \\
    Cost matrices & $Q = 10^3 I_p$, $R = I_m$ \\
    Safety constraint coefficients & $E = I_{m+p} \otimes \smallmat{1\\-1}$ \\
    & $f = [.1\; .1\; .5\; .1\; .4\; .4\; .4\; .4]^\transpose$ \\
    CVaR level$^{\rm a}$ & $\alpha = 0.3$ \\
    Variance of $v_t$ for SMPC/SDDPC & $\Sigma^{\sf v} = 5 \times 10^{-7} I_p$ \\
    Variance of $\rho_t$ for SDDPC & $\Sigma^\rho = 10^{-7} I_{pL}$ \\
    Variance of $w_t$ for SMPC$^{\rm b}$ & $\Sigma^{\sf w} = \mathcal{O}^\dagger \Sigma^\rho \mathcal{O}^{\dagger \transpose}$ \\ \bottomrule 
    \multicolumn{2}{@{$^{\rm a}$}l@{}}{$\alpha$ is used as the risk bound for chance constrained controllers.} \\
    \multicolumn{2}{@{$^{\rm b}$}l@{}}{$\mathcal{O}$ is obtained given the identified model $(A,B,C,D)$ in SMPC.}
\end{tabular} \end{table}

\begin{table}[t]
\caption{Simulation Result Statistics} \label{TABLE:simulation_results}
\centering 
\begin{tabular}{@{}lccc@{}} 
    \toprule
    \multirow{2}{*}{\textbf{Controller}} & \multicolumn{2}{c}{\textbf{Total Tracking Cost}} & \textbf{Cumulative Violation} \\ 
    & $0$s to $30$s & $30$s to $60$s & from $60$s to $90$s \\ \midrule
    \textbf{DR/O-SDDPC}$^{\rm a}$ & $0.02$ & $64.2$ & $0$ \\
    DR/F-SDDPC & $0.02$ & $68.9$ & $0$ \\
    CC/F-SDDPC & $0.02$ & $64.9$ & $0.03$ \\ \midrule
    DR/O-SMPC & $0.02$ & $64.2$ & $0$ \\
    DR/F-SMPC & $0.02$ & $68.0$ & $0$ \\
    CC/F-SMPC & $0.02$ & $64.9$ & $0.01$ \\ \midrule
    deterministic MPC & $0.09$ & $64.6$ & $0.20$ \\
    SPC & $0.18$ & $65.5$ & $2.23$ \\
    DeePC & $0.18$ & $64.7$ & $0.19$ \\
    \bottomrule 
    \multicolumn{4}{@{$^{\rm a}$}l@{}}{DR -- distributionally robust constrained, CC -- chance constrained,} \\
    \multicolumn{4}{@{\;\;}l@{}}{O -- with optimized feedback gain, F -- with fixed feedback gain.}
\end{tabular} \end{table}
}}

\section{Conclusions} \label{SECTION:conclusions}

We proposed a Stochastic Data-Driven Predictive Control (SDDPC) method that accommodates distributionally robust (DR) probability constraints and produces closed-loop control policies with feedback gains determined from optimization. In theory, our SDDPC method can produce equivalent control inputs with associated Stochastic MPC, under specific conditions. Simulation results indicated separate benefits of using DR constraints and optimized feedback gains.

\setlength{\abovedisplayskip}{.2em}
\setlength{\belowdisplayskip}{.2em}
\setcounter{section}{0}
\renewcommand{\thesection}{Appendix \Alph{section}}
\section{Definition of $\Delta^{\sf U}_i$, $\Delta^{\sf Y}_i$, $\Delta^{\sf A}_i$, $\Delta^{\sf M}$} \label{APPENDIX:Delta_Lambda_definition}

The matrices $\Delta^{\sf U}_i \in \real^{m \times mN}, \Delta^{\sf Y}_i \in \real^{p \times mN}, \Delta^{\sf A}_i \in \real^{p \times n_\eta}$ for $i \in \integer_{[0,N)}$ and $\Delta^{\sf M} \in \real^{pN \times n_\eta}$ in \eqref{Eq:Lambda_definition} are what follows,
\begin{align*} \begin{aligned}
    & \col \big( \Delta^{\sf U}_0, \ldots, \Delta^{\sf U}_{N-1} \big) := I_{mN} \\
    & \col \big( \Delta^{\sf Y}_0, \ldots, \Delta^{\sf Y}_{N-1} \big) := \Xi(A) \, (I_N \otimes B) \\
    & \col \big( \Delta^{\sf A}_0, \ldots, \Delta^{\sf A}_{N-1} \big) := [\Theta(A), \Xi(A), I_{pN}] \\
    &\Delta^{\sf M} := [\Theta(A_{\sf L}), \Xi(A_{\sf L}), I_{pN} - \Xi(A_{\sf L}) \, (I_N \otimes L_{\sf L})]
\end{aligned} \end{align*}
where we let $\Theta(A) := \col(C, CA, \ldots, CA^{N-1}) \in \real^{pN \times n}$, $\Xi(A) := \Toep(0_{p\times n}, C, CA, \ldots, CA^{N-2}) \in \real^{pN \times nN}$, and similarly define $\Theta(A_{\sf L}), \Xi(A_{\sf L})$ with $A_{\sf L} := A - L_{\sf L} C$.

The matrices $\mathbf{\Delta}^{\sf U}_i, \mathbf{\Delta}^{\sf Y}_i, \mathbf{\Delta}^{\sf A}_i, \mathbf{\Delta}^{\sf M}$ in \eqref{Eq:DD:Lambda_definition} are computed (with underlying $\mathbf{\Theta}(\mathbf{A}), \mathbf{\Xi}(\mathbf{A}), \mathbf{A}_{\sf L}$) in the same way as above, with $A, B, C, L_{\sf L}, n$ replaced by $\mathbf{A}, \mathbf{B}, \mathbf{C}, \mathbf{L}_{\sf L}, n_\auxiliary$, respectively.

\Versions{}{{\tgrey \section{Proof of \eqref{Eq:PROOF:equivalence:target}} \label{APPENDIX:proof}

\begin{proof}
The relation $\thickbar y_t = \thickbar{\mathbf{y}}_t$ in \eqref{Eq:PROOF:equivalence:target} was established in \cite[Claim 7.7]{SDDPC}. The other relation in \eqref{Eq:PROOF:equivalence:target} is equivalent to
\begin{align*} \begin{aligned}
    \Delta^{\sf U}_{t-k} &= \mathbf\Delta^{\sf U}_{t-k}, &
    \Delta^{\sf M} \Sigma^\eta (\Delta^{\sf M})^\transpose &= \mathbf\Delta^{\sf M} \mathbf\Sigma^\eta (\mathbf\Delta^{\sf M})^\transpose, \\
    \Delta^{\sf Y}_{t-k} &= \mathbf\Delta^{\sf Y}_{t-k}, &
    \Delta^{\sf A}_{t-k} \Sigma^\eta (\Delta^{\sf A}_{t-k})^\transpose &= \mathbf\Delta^{\sf A}_{t-k} \mathbf\Sigma^\eta (\mathbf\Delta^{\sf A}_{t-k})^\transpose,
\end{aligned} \end{align*}
via the definitions of $\Lambda_t$ and $\mathbf{\Lambda}_t$ in \eqref{Eq:Lambda_definition} and \eqref{Eq:DD:Lambda_definition}. Given the definitions in \ref{APPENDIX:Delta_Lambda_definition}, the above relations are implied by
\begin{enumerate}[1)]
    \item $C A^q B = \mathbf{C} \mathbf{A}^q \mathbf{B}$ for $q \in \integer_{[0,N)}$,
    \item $C A^q \Sigma^{\sf x} (C A^r)^\transpose = \mathbf{C} \mathbf{A}^q \mathbf{\Sigma}^{\sf x} (\mathbf{C}\mathbf{A}^r)^\transpose$ for $q,r \in \integer_{[0,N)}$,
    \item $C A^q \Sigma^{\sf w} (C A^r)^\transpose = \mathbf{C} \mathbf{A}^q \mathbf{\Sigma}^{\sf w} (\mathbf{C}\mathbf{A}^r)^\transpose$ for $q,r \in \integer_{[0,N-2]}$,
    \item $C A_{\sf L}^q \Sigma^{\sf x} (C A_{\sf L}^q)^\transpose = \mathbf{C} \mathbf{A}_{\sf L}^q \mathbf{\Sigma}^{\sf x} (\mathbf{C} \mathbf{A}_{\sf L}^r)^\transpose$ for $q,r \in \integer_{[0,N)}$,
    \item $C A_{\sf L}^q \Sigma^{\sf w} (C A_{\sf L}^q)^\transpose = \mathbf{C} \mathbf{A}_{\sf L}^q \mathbf{\Sigma}^{\sf w} (\mathbf{C} \mathbf{A}_{\sf L}^r)^\transpose$ for $q,r \in \integer_{[0,N-2]}$,
    \item $C A_{\sf L}^q L_{\sf L} = \mathbf{C} \mathbf{A}_{\sf L}^q \mathbf{L}_{\sf L}$ for $q \in \integer_{[0,N-2]}$,
\end{enumerate}
where the relations 1)--6) can be shown given the equalities
\begin{align*} \begin{aligned}
    A \Phi \Phi_\auxiliary &= \Phi \mathbf{A} \Phi_\auxiliary, &
    B &= \Phi \mathbf{B}, &
    L_{\sf L} &= \Phi \mathbf{L}_{\sf L}, \\
    \mathbf{A} \Phi_\auxiliary &= \Phi_\auxiliary \widetilde A, &
    \mathbf{B} &= \Phi_\auxiliary \widetilde B, & 
    \mathbf{L}_{\sf L} &= \Phi_\auxiliary \widetilde L_{\sf L}, \\
    C \Phi \Phi_\auxiliary &= \mathbf{C} \Phi_\auxiliary, &
    \Sigma^{\sf w} &= \Phi \mathbf{\Sigma}^{\sf w} \Phi^\transpose, &
    \Sigma^{\sf x} &= \Phi \mathbf{\Sigma}^{\sf x} \Phi^\transpose,
\end{aligned} \end{align*}
established with some matrices $\Phi, \widetilde A, \widetilde B, \widetilde L_{\sf L}$ according to \cite[Claim 5.1, Claim 7.1, Claim 7.4]{SDDPC}.
\end{proof} }}

\end{document}